\documentclass[12pt,english]{article}
\usepackage[utf8]{inputenc}
\usepackage[english]{babel}
\usepackage{float,caption,subcaption}
\usepackage{booktabs,longtable}
\usepackage{amsmath}
\usepackage{amssymb,amsthm,amsfonts}
\usepackage{enumerate}
\usepackage{tikz}
\usepackage[authoryear,round]{natbib}
\usepackage{bibentry}
\usepackage[hyphens]{url}
\usepackage{hyperref}
\usepackage[a4paper,body={16cm,23.7cm}]{geometry}

\hypersetup{colorlinks=true, linkcolor=cyan, citecolor=cyan, urlcolor=black, breaklinks=true}
\newtheorem{theorem}{Theorem}

\newtheorem{proposition}{Proposition}
\newtheorem{remark}{Remark}

\let\emptyset\varnothing
\DeclareMathOperator{\NP}{NP}
\DeclareMathOperator{\NfP}{NfP}
\newcommand{\V}{\mathcal{V}}

\begin{document}

\title{\textbf{NULL PLAYER NEUTRALITY IN TU-GAMES: EGALITARIAN AND SHAPLEY SOLUTIONS}\thanks{Juan Carlos Gon\c{c}alves-Dosantos acknowledges the grants PID2021-12403030NB-C31 and PID2022-137211NB-I00 funded by MCIN/AEI/10.13039/501100011033 and by ``ERDF A way of making Europe/EU''. Ricardo Mart\'inez acknowledges the R\&D\&I project grant PID2023-147391NB-I00 funded by MCIN AEI/10.13039/501100011033 and by ``ERDF A way of making Europe/EU''. Joaqu\'in S\'anchez-Soriano acknowledges the grant PID2022-137211NB-I00 funded by MCIN/AEI/10.13039/501100011033 and by ``ERDF A way of making Europe/EU'' and the CIPROM/2024/34 grant, funded by the Conselleria de Educación, Cultura, Universidades y Empleo, Generalitat Valenciana.}}
\author{Juan Carlos Gon\c{c}alves-Dosantos\thanks{email: jgoncalves@umh.es} \\ Universidad Miguel Hern\'andez de Elche \and Ricardo Mart\'inez \\ Universidad de Granada \and Joaqu\'in S\'anchez-Soriano \\ Universidad Miguel Hern\'andez de Elche}
\date{}
\maketitle

\begin{abstract}
We introduce and study the axiom of \emph{null player neutrality} in the context of cooperative games with transferable utility (TU-games). This axiom weakens the classical coalitional strategic equivalence: rather than requiring that augmenting a game by a null-player game leaves that player's payoff unchanged, it only requires that any change in payoff be independent of the specific augmenting game, provided both the null-player condition and the grand-coalition value are preserved. We show that efficiency, linearity, symmetry, and null player neutrality together characterize the family of all real linear combinations of the Shapley value and the equal division solution, a family that strictly extends the well-known class of $\alpha$-egalitarian Shapley values (convex combinations, $\alpha \in [0,1]$) to arbitrary $\alpha \in \mathbb{R}$. Replacing null player neutrality by its natural analogue for nullifying players uniquely pins down the equal division solution. 
\end{abstract}

\textbf{Keywords}: TU-games, null player property, null player neutrality, nullifying player neutrality, Shapley solution, Equal division solution.

\textbf{JEL Codes}: C71, D71

\newpage

\section{Introduction}

Cooperative game theory with transferable utility (TU-games) provides a rigorous analytical framework for the study of environments in which groups of agents generate value through collective action, as well as for the analysis of how this value may be allocated among participants. Since the foundational contributions of the 1950s, the literature has been largely structured around two polar allocation principles. At one extreme lies the \emph{Shapley value} \citep{shapley1953value}, which assigns to each player a payoff equal to her expected marginal contribution across all possible coalitions she may join, thereby embodying a strictly marginalist perspective. At the opposite extreme stands the \emph{equal division solution}, which allocates the worth of the grand coalition uniformly among all players, reflecting a fully egalitarian principle that is entirely insensitive to individual productivity or contribution. Despite their philosophical differences, both solutions satisfy the fundamental structural axioms of \emph{efficiency}, \emph{linearity}, and \emph{symmetry}, and both have been widely axiomatized and applied in a wide range of fields, including Economics, Political Science, and Operations Research. The Shapley value, in particular, has been the object of a rich body of axiomatic characterizations. Beyond Shapley's original axiomatization via the null player property, \citet{Young85} derives it from strong monotonicity; \citet{Hart89} obtain it through a potential-function approach that extends naturally to non-transferable utility games; and \citet{Casajus17b} characterize it via differential marginality, which requires that the payoff difference between two players track the difference in their marginal contributions. By contrast, the equal division solution has received less systematic axiomatic attention; prominent characterizations are due to \citet{Brink07}. Between these extremes lies the one-parameter family of \emph{egalitarian Shapley values} \citep{Joosten96}, which interpolate between marginalism and egalitarianism.

The main contribution of this paper is the introduction and systematic analysis of a new axiom, termed \emph{null player neutrality}, which offers an alternative to the classical \emph{null player property} and its variants. The null player property, a cornerstone of Shapley's original axiomatization, requires that any player who contributes nothing to every coalition receive a zero payoff. While compelling from a strict marginalist perspective, it embeds nontrivial normative assumptions that deserve careful scrutiny, as they may not be universally acceptable. Taken at face value, the null player property enshrines the principle that only productive contribution justifies a claim on collectively generated surplus, thereby excluding from any share of social output those agents who, for whatever reason, happen to make no measurable marginal contribution. This perspective may sit uncomfortably with some ethical positions. From a Rawlsian standpoint \citep{Rawls71}, for example, the claim of the least advantaged members of a society should not depend on their productive output; the null player property, by contrast, offers such members nothing. More broadly, the property is in tension with welfarist and needs-based principles of distributive justice, which hold that the grounds for a share of collective wealth may include need, vulnerability, or mere membership in the cooperative enterprise, rather than marginal productivity alone \citep{Sen80}. These concerns also arise naturally in applied contexts, where the null player property would, for instance, legitimize the exclusion of unemployed individuals from shares of national income or social security benefits. In cooperative game theory, related concerns have motivated solidarity-based principles designed to protect less favored individuals.  \citet{Nowak94} introduce the solidarity value, which replaces each player's marginal contribution with the average marginal contribution of her coalition, thereby ensuring that null players are compensated whenever their partners are productive. The broader conflict between egalitarianism and marginalism in cooperative environments is analyzed by \citet{Dutta89}, who characterize the egalitarian solution under participation constraints.

The axiom we propose, \emph{null player neutrality}, relaxes the requirement imposed by the null player property. Building on the augmentation framework underlying \emph{coalitional strategic equivalence} \citep{Chun91}, it considers the addition of games in which a given player is null, but requires only that any resulting change in that player's payoff be independent of the particular null game used, provided that the grand coalition's worth remains fixed. Formally, for any three games $v,w,u$ such that $w(N) = u(N)$ and player $i$ is null in both $w$ and $u$, null player neutrality requires $\varphi_i(v+w) = \varphi_i(v+u)$. This axiom is less demanding than both the null player property and coalitional strategic equivalence. In addition, it is satisfied by the Shapley value, the equal division solution, and all convex combinations thereof, as well as by a substantially broader class of allocation schemes, none of which are reachable through the null player property or coalitional strategic equivalence alone.

Our first main result (Theorem \ref{thm1}) shows that efficiency, linearity, symmetry, and null player neutrality together characterize precisely the family $\{\alpha\varphi^{ED}+(1-\alpha)\varphi^{Sh}:\alpha\in\mathbb{R}\}$. This is a direct generalization of Shapley's classical characterization: weakening the null player property to null player neutrality, while preserving the remaining structural axioms, enlarges the set of admissible solutions from a single point to a one-parameter family. The subfamily corresponding to $\alpha\in[0,1]$ recovers the $\alpha$-egalitarian Shapley values of \citet{Joosten96}, axiomatized by \citet{Brink13} using \emph{weak monotonicity}. Our characterization replaces that condition with null player neutrality, and thereby extends the admissible family to all $\alpha\in\mathbb{R}$. We also compare our result with the characterization of \citet{Casajus13}, who use the \emph{null player in a productive environment} property to pin down the subfamily with $\alpha \geq 0$. We show that, under linearity, their axiom implies null player neutrality but not vice versa, explaining the strict enlargement of the family obtained in Theorem \ref{thm1}.

We also provide a dual analysis for \emph{nullifying players} (players whose presence reduces every coalition's value to zero \citep{Brink07}). Parallel to null player neutrality, we introduce the property of \emph{nullifying player neutrality}, which requires that a nullifying player's payoff remain unchanged across games with the same grand-coalition value. Unlike the null player case, this axiom is equivalent to coalitional standard equivalence \citep{Brink07}, since the zero game is nullifying for every player. Under linearity, nullifying player neutrality is also equivalent to the nullifying player property, in contrast to the weaker relationship found for null players. We show that efficiency, symmetry, and nullifying player neutrality uniquely characterize the equal division solution (Theorem \ref{thm2}). Together with Theorem \ref{thm1}, this yields a clear dichotomy. Null player neutrality permits all linear combinations of the Shapley and equal division solutions, whereas nullifying player neutrality singles out the latter alone.

The remainder of the paper is organized as follows. Section \ref{prel} introduces the notation and definitions. Section \ref{NPN} presents null player neutrality, establishes its logical relationships with related axioms, and proves our main characterization result (Theorem \ref{thm1}). Section \ref{sec_further} develops the parallel analysis for nullifying players. Section \ref{sec_conclusion} collects concluding remarks and directions for future research.

\section{Preliminaries}\label{prel}

A cooperative game with transferable utility (a TU-game) models situations in which a finite set of players can generate value through cooperation and this value can be freely transferred among them. Formally, a \textbf{TU-game} is the pair $(N, v)$, where $N$ is a finite set of \textbf{players} (we assume $|N| \geq 3$) and $v : 2^N \to \mathbb{R}$ is the \textbf{characteristic function} satisfying $v(\emptyset) = 0$. Subsets of $N$ are called \textbf{coalitions} and $N$ is the grand coalition. For each $S \subseteq N$, the worth of coalition $S$ is $v(S)$, and we write $s = |S|$ and $n = |N|$. The class of all TU-games on player set $N$ is denoted by $\V^N$. When no ambiguity arises, we identify $(N, v)$ simply by $v$. \footnote{For notational convenience, we write $i$ in place of the singleton $\{i\}$ whenever no ambiguity arises. Accordingly, the expressions $v(i)$, $v(S \cup i)$, and $v(S \setminus i)$ are equivalent to $v(\{i\})$, $v(S \cup \{i\})$, and $v(S \setminus \{i\})$, respectively.}

Two particular classes of TU-games will be useful throughout. Given $T \subseteq N$, $T \neq \emptyset$, the \emph{unanimity game} $u_T$ is defined by $u_T(S) = 1$ if $T \subseteq S$, and $u_T(S) = 0$ otherwise. The \emph{canonical game} $e_T$ is defined by $e_T(S) = 1$ if $T = S$, and $e_T(S) = 0$ otherwise. Both families form standard bases of $\V^N$: for every $v \in \V^N$ there exist unique families of coefficients $(\lambda_T)_{T \subseteq N,\,T\neq\emptyset}$ and $(\gamma_T)_{T \subseteq N,\,T\neq\emptyset}$ such that $v = \sum_{T \subseteq N,\,T\neq\emptyset} \lambda_T u_T = \sum_{T \subseteq N,\,T\neq\emptyset} \gamma_T e_T$.

A \textbf{solution} is a mapping $\varphi : \V^N \to \mathbb{R}^n$ that assigns to each game $v$ a payoff vector $\varphi(v)$ where the component $\varphi_i(v)$ indicates the player $i$'s payoff.

Among the many solutions studied in the literature, two
stand out as opposite focal points that have anchored the debate due to their fundamentally contrasting approaches. The \emph{Shapley solution} \citep{shapley1953value} is grounded in the idea that each player's payoff should reflect their productive importance to the coalition structure. It computes, for each player, the average marginal contribution across all possible orders in which players might join the grand coalition, weighting each order equally. The \emph{equal division solution}, at the opposite extreme, ignores individual contributions entirely and divides the worth of the grand coalition equally among all players.

\textbf{Shapley solution}. For each $v \in \V^N$ and each $i \in N$,
$$
\varphi^{Sh}_i(v) = \sum_{S \subseteq N : i \in S}  \frac{(s-1)!(n-s)!}{n!} \left( v(S) - v(S \backslash i) \right).
$$

\textbf{Equal division solution}. For each $v \in \V^N$ and each $i \in N$,
$$
\varphi^{ED}_i(v) = \frac{v(N)}{n}.
$$

The \emph{$\alpha$-egalitarian Shapley values}, introduced by \citet{Joosten96}, are obtained as convex combinations of the Shapley  and the equal division solutions, indexed by a parameter $\alpha \in [0,1]$ that interpolates between the two polar principles. 

\textbf{$\alpha$-egalitarian Shapley solution}. For each $\alpha \in [0,1]$, and each $v \in \V^N$,
$$
\varphi^\alpha(v) = \alpha \varphi^{ED}(v) + (1-\alpha) \varphi^{Sh} (v).
$$



\section{Null player neutrality}
\label{NPN}

Several axioms have been employed in the literature to characterize the solutions discussed above. We present those most relevant to our analysis, starting with standard structural requirements and then turning to the axiom introduced in this paper.

\emph{Efficiency} states that the sum of the players' payoffs equals the worth of the grand coalition.

\textbf{Efficiency}. For each $v \in \V^N$,
$$
\sum_{i \in N} \varphi_i(N,v) = v(N).
$$

The next property ensures that payoffs respond in a consistent way to scalar transformations and combinations of games.\footnote{Linearity implies the classical \emph{additivity} property, which requires that the payoff assigned to each player in the sum of two games coincides with the sum of that player’s payoffs in the two separate games, i.e., for each $v,w \in \V^N$, $\varphi(v+w)=\varphi(v)+\varphi(w)$.}

\textbf{Linearity}. For each $v,w \in \V^N$ and $a,b\in \mathbb{R}$,
$$
\varphi(av+bw) = a\varphi(v) + b\varphi(w),
$$
where $(av+bw)(S)=av(S)+bw(S)$ for any $S \subseteq N$.

Two players $\{i,j\} \subseteq N$ are \emph{symmetric} in $v$ if they contribute the same to any coalition, that is, if $v(S \cup i) = v(S \cup j)$ for any $S \subseteq N \backslash \{i,j\}$. The following axiom embodies a basic fairness requirement, namely, that symmetric players receive identical payoffs.

\textbf{Symmetry}. For each $v \in \V^N$, if $\{i,j\} \subseteq N$ are symmetric, then $\varphi_i(v)=\varphi_j(v)$.

The following properties concern the treatment of players whose contributions to the value generated through cooperation are negligible. A player $i \in N$ is \emph{null} if they contribute nothing to any coalition, i.e., $v(S \cup i)=v(S)$ for all $S \subseteq N \setminus \{i\}$. For any $v \in \V^N$, $\NP(v)$ denotes the set of null players in $v$. The next axiom requires that null players receive zero.

\textbf{Null player property}. For each $v \in \V^N$ and each $i \in N$, if $i$ is null then $\varphi_i(v)=0$. 

A further principle addressing the implications of being a null player was introduced by \cite{Chun89}, who proposed the axiom of \emph{coalitional strategic equivalence}. It states that if a game is augmented by another game in which a given player is null, then the payoff assigned to that player should remain unchanged.

\textbf{Coalitional strategic equivalence}. For each $v,w \in \V^N$ and each $i \in N$, if $i$  is null in $w$, then $\varphi_i(v+w) = \varphi_i(v)$.

This axiom therefore says that \emph{null-player games} $w$ carry no information for the payoff of the null players. Such games may be viewed as irrelevant perturbations of the original game $v$ from the perspective of null agents. Coalitional strategic equivalence is satisfied by the Shapley solution (as a consequence of the null player property combined with linearity) but is violated by the equal division solution. Hence, the condition is too strong to encompass egalitarian solution concepts.

The next principle, \emph{null player neutrality}, preserves the core insight of coalitional strategic equivalence (that the internal structure of a null-player game should not affect the payoff of a null player) while dropping the requirement that the augmentation has no effect at all. Specifically, the axiom requires that, for any two games in which player $i$ is null and that have the same grand coalition worth, the payoff of $i$ must be the same after augmentation. In other words, a null player is indifferent between any two augmentations that generate the same aggregate value. That is, the internal distribution of worth across coalitions in $w$ and $u$, beyond what reaches the grand coalition, is irrelevant to the player's payoff.

\textbf{Null player neutrality}. For each $v,w,u \in \V^N$ and each $i \in N$, if $w(N)=u(N)$ and $i$ is null in both $w$ and $u$, then $\varphi_i(v+w) = \varphi_i(v+u)$.

The requirement $w(N)=u(N)$ isolates the only channel through which an augmentation may legitimately affect the payoff of a null player. Even when player $i$ is null in both $w$ and $u$, these games may change the total worth available for allocation. Since any efficient solution distributes $v(N)$, changes in $w(N)$, or $u(N)$, may justifiably affect the payoff of all players, including null players.

By restricting attention to perturbations that generate the same worth of the grand coalition, the axiom abstracts from this aggregate effect and focuses exclusively on the internal structure of the games.

It is straightforward to observe that coalitional strategic equivalence implies null player neutrality. Indeed, let $v,w,u \in \V^N$ be such that $w(N) = u(N)$. If $i \in N$ is a null player in both $w$ and $u$, then, by applying coalitional strategic equivalence, we have $\varphi_i(v+w) = \varphi_i(v) = \varphi_i(v+u)$. The converse implication does not hold, as the equal division solution satisfies null player neutrality but fails to satisfy coalitional strategic equivalence. Besides, the Shapley solution also satisfies the axiom, as for a null player $i$, $\varphi_i^{Sh}(w) = 0$ for any game $w$ in which $i$ is null, so changing $w$ to $u$ with $w(N) = u(N)$ does not affect $\varphi_i^{Sh}(v+w) = \varphi_i^{Sh}(v) = \varphi_i^{Sh}(v+u)$. Null player neutrality thus carves out a natural middle ground: it imposes discipline on how solutions treat null players under perturbations, without forcing a specific point allocation on them.

Remark \ref{remark_NPPandNPN} shows that the null player property and null player neutrality are logically independent.

\begin{remark}
\label{remark_NPPandNPN}
The null player property and null player neutrality are independent.
\begin{itemize}
  \item[(i)] The solution $\varphi^1$, defined as follows,
  $$
  \varphi^1_i(v) = 
  \begin{cases}
    \frac{v(i)}{\sum_{j \in N} v(j)} \cdot v(N) & \text{if } \sum_{j \in N} v(j) \neq 0 \\[0.2cm]
    0 & \text{if } \sum_{j \in N} v(j) = 0 \text{ and } i \in \NP(v) \\[0.2cm]
    \frac{v(N)}{|N \backslash \NP(v)|} & \text{if } \sum_{j \in N} v(j) = 0 \text{ and } i \notin \NP(v),
  \end{cases}
  $$
  satisfies the null player property but violates null player neutrality. To see this, consider the games $v, w, u$ on $N = \{1,2,3\}$ whose characteristic functions are specified below:
  \begin{center}
  \begin{tabular}{cccc}
    \toprule
    $S$ & $v(S)$ & $w(S)$ & $u(S)$ \\
    \midrule
    $\varnothing$ & 0 & 0 & 0 \\
    $\{1\}$ & 1 & 0 & 0 \\
    $\{2\}$ & 1 & 2 & 0 \\
    $\{3\}$ & 1 & 0 & 0 \\
    $\{1,2\}$ & 0 & 2 & 0 \\
    $\{1,3\}$ & 0 & 0 & 0 \\
    $\{2,3\}$ & 0 & 2 & 2 \\
    $\{1,2,3\}$ & 2 & 2 & 2 \\
    \bottomrule
  \end{tabular}
  \end{center}
  Player 1 is null in both $w$ and $u$, and $w(N) = u(N) = 2$. However, $\varphi^1_1(v+w) \neq \varphi^1_1(v+u)$, violating null player neutrality.

   \item[(ii)] The equal division solution satisfies null player neutrality but violates the null player property.
\end{itemize}
\end{remark}

Although the null player property and null player neutrality are independent in general, linearity creates a bridge: any linear solution satisfying the null player property also satisfies null player neutrality. The converse, however, does not hold. In particular, the equal division solution $\varphi^{ED}$ is linear and satisfies null player neutrality, yet it violates the null player property. This shows that, under linearity, null player neutrality is a weaker requirement than the null player property.

The following result is well known.

\begin{theorem}[\cite{shapley1953value}]
A solution $\varphi$ satisfies efficiency, linearity, symmetry, and the null player property if and only if it is the Shapley solution, $\varphi \equiv \varphi^{Sh}$.
\end{theorem}

We now present our main characterization. In essence, we explore the implications of replacing the null player property with null player neutrality in the preceding result. Our first theorem identifies the class of solutions that satisfy efficiency, linearity, symmetry, and null player neutrality. We show that this class coincides precisely with the set of linear combinations of the equal division and Shapley solutions.

\begin{theorem}
\label{thm1}
A solution $\varphi$ satisfies efficiency, linearity, symmetry, and null player neutrality if and only if there exists $\alpha \in \mathbb{R}$ such that
$$
\varphi \equiv \alpha \varphi^{ED} + (1-\alpha) \varphi^{Sh}.
$$
\end{theorem}
\begin{proof}
We first verify that any solution of the stated form satisfies all four properties.
Efficiency, linearity, and symmetry hold because both $\varphi^{Sh}$ and $\varphi^{ED}$
satisfy them and are preserved under linear combinations. To verify null
player neutrality, let $v,w,u\in\V^N$ and suppose player $i$ is null in both
$w$ and $u$ with $w(N)=u(N)$. Since $i$ is null in $w$ and in $u$,
$\varphi^{Sh}_i(w)=\varphi^{Sh}_i(u)=0$. Moreover, since $w(N)=u(N)$,
$\varphi^{ED}_i(w)=\varphi^{ED}_i(u)$. Then,
\begin{align*}
  \varphi_i(v+u)
    &= \alpha \varphi_i^{ED}(v+u) + (1-\alpha) \varphi_i^{Sh}(v+u) \\
    &= \alpha \bigl(\varphi_i^{ED}(v)+\varphi_i^{ED}(u)\bigr)
       + (1-\alpha) \bigl(\varphi_i^{Sh}(v)+\varphi_i^{Sh}(u)\bigr)\\
    &= \alpha \bigl(\varphi_i^{ED}(v)+\varphi_i^{ED}(u)\bigr)
       + (1-\alpha)\,\varphi_i^{Sh}(v)\\
    &= \alpha \bigl(\varphi_i^{ED}(v)+\varphi_i^{ED}(w)\bigr)
       + (1-\alpha)\,\varphi_i^{Sh}(v)\\
    &= \alpha \bigl(\varphi_i^{ED}(v)+\varphi_i^{ED}(w)\bigr)
       + (1-\alpha) \bigl(\varphi_i^{Sh}(v)+\varphi_i^{Sh}(w)\bigr)\\
    &= \varphi_i(v+w).
\end{align*}
 
We now prove necessity. Let $\varphi$ be a solution satisfying the four axioms,
and let $v \in \V^N$. Since the unanimity games $\{u_T\}_{T \subseteq N,\,T\neq\emptyset}$
form a basis of $\V^N$, there exist coefficients $(\lambda_T)_{T \subseteq N,\,T\neq\emptyset}$
such that $v = \sum_{T \subseteq N,\,T\neq\emptyset} \lambda_T\, u_T$. By \emph{linearity},
$$
\varphi(v) = \sum_{T \subseteq N,\,T\neq\emptyset} \lambda_T\,\varphi(u_T).
$$
 
We first show that, for any nonempty subsets $T,T' \subseteq N$, any $i \notin T$, and any $j \notin T'$, we have
\[
\varphi_i(u_T) = \varphi_j(u_{T'}).
\]

Note that $i \in N \setminus T$ if and only
if $i$ is null in $u_T$. We distinguish two cases.
 
\begin{itemize}
  \item[(i)] $T \cup T' \subsetneq N$. Then there exists $l \in N$ with $l \notin T$
  and $l \notin T'$, i.e., $l \in \NP(u_T) \cap \NP(u_{T'})$. Since $u_T(N)=u_{T'}(N)=1$, by \emph{null player neutrality}, $\varphi_l(w + u_T) = \varphi_l(w + u_{T'})$ for
  any $w \in \V^N$. Applying \emph{linearity} to both sides and cancelling
  $\varphi_l(w)$ yields $\varphi_l(u_T) = \varphi_l(u_{T'})$. Since $i$ and $l$ are
  both null in $u_T$, and $j$ and $l$ are both null in $u_{T'}$, \emph{symmetry}
  gives $\varphi_i(u_T) = \varphi_l(u_T) = \varphi_l(u_{T'}) = \varphi_j(u_{T'})$.
 
  \item[(ii)] $T \cup T' = N$. Since $T \cup T' = N$ and $|N| \geq 3$, at least one
of the two coalitions contains at least two players. Without loss of generality,
assume that $|T| \geq 2$. Moreover, since $j \notin T'$ and $T \cup T'=N$, we have
$j \in T$. Hence, there exists $k \in T \setminus \{j\}$. Set $\hat{T}=\{k\}$.
Then $\hat{T}$ is nonempty, $\hat{T}\subseteq T$, and $j\notin \hat{T}$.  Since $i \notin T$, player $i$ is null in both $u_T$
  and $u_{\hat{T}}$ and $u_T(N)=u_{\hat{T}}(N)=1$, so by \emph{null player
  neutrality} and \emph{linearity}, $\varphi_i(u_T) = \varphi_i(u_{\hat{T}})$.
  Analogously, since $j \notin T'$ and $j \notin \hat{T}$, we obtain
  $\varphi_j(u_{T'}) = \varphi_j(u_{\hat{T}})$. Since both $i$ and $j$ are null in
  $u_{\hat{T}}$, by \emph{symmetry}
  $\varphi_i(u_{\hat{T}}) = \varphi_j(u_{\hat{T}})$. Combining yields
  $\varphi_i(u_T) = \varphi_j(u_{T'})$.
\end{itemize}
 
Now let $T \subseteq N$, with $T\neq \emptyset$, and consider the unanimity game $u_T$. If
$\{i,j\} \subseteq T$ or $\{i,j\} \subseteq N \setminus T$, then $i$ and $j$ are
symmetric in $u_T$. Consequently, the players of $u_T$ split into two groups: $T$,
whose members are non-null and pairwise symmetric, and $N \setminus T$, whose members
are null and pairwise symmetric. By the argument proved above, the common payoff of
null players does not depend on $T$; call it $\eta \in \mathbb{R}$. That is,
$\varphi_i(u_T) = \eta$ for every $i \in N \setminus T$. By \emph{efficiency},
$\sum_{k \in N} \varphi_k(u_T) = u_T(N) = 1$, and by \emph{symmetry} all players in
$T$ receive the same payoff, so $\varphi_i(u_T) = \frac{1-(n-t)\eta}{t}$ for every
$i \in T$.
 
Set $\alpha := n\eta$. We verify that
$\varphi_i(u_T) = \alpha\,\varphi_i^{ED}(u_T) + (1-\alpha)\,\varphi_i^{Sh}(u_T)$
in each case. Recall that, for unanimity games $u_T$, the Shapley value and the equal division value are given by
\[
\varphi_i^{Sh}(u_T)=
\begin{cases}
\dfrac{1}{t}, & \text{if } i\in T\\
0, & \text{if } i\notin T,
\end{cases}
\qquad
\varphi_i^{ED}(u_T)=\frac{1}{n}
\quad \text{for all } i\in N.
\]

Therefore,
\begin{itemize}
  \item[(i)] If $i \in N \setminus T$ ($i$ is null in $u_T$):
  $$
  \alpha\,\varphi_i^{ED}(u_T) + (1-\alpha)\,\varphi_i^{Sh}(u_T)
  = \alpha\frac{1}{n} + (1-\alpha) 0
  = n\eta\frac{1}{n} = \eta = \varphi_i(u_T).
  $$
  \item[(ii)] If $i \in T$ ($i$ is non-null in $u_T$):
  $$
  \alpha\,\varphi_i^{ED}(u_T) + (1-\alpha)\,\varphi_i^{Sh}(u_T)
  = \alpha\frac{1}{n} + (1-\alpha)\frac{1}{t}
  = \eta + \left(1-n\eta\right)\frac{1}{t}
  = \frac{1-(n-t)\eta}{t}
  = \varphi_i(u_T).
  $$
\end{itemize}
 
Finally, by \emph{linearity},
$$
\varphi(v) = \alpha\,\varphi^{ED}(v) + (1-\alpha)\,\varphi^{Sh}(v). 
$$
\end{proof}

Theorem~\ref{thm1} can be understood as a direct generalization of Shapley's classical characterization. As highlighted above, under linearity, null player neutrality constitutes a weakening of the null player property. This result shows that relaxing the latter condition, while preserving the remaining structural axioms, enlarges the set of admissible solutions from a single point to a one-parameter family. At the same time, Theorem~\ref{thm1} illustrates the normative role of the null player property within this family as it is precisely this requirement that forces the parameter~$\alpha$ to collapse to zero.

The $\alpha$-egalitarian Shapley values, introduced by \citet{Joosten96} and axiomatized by \citet{Brink13}, constitute the subfamily of our family corresponding to $\alpha \in [0,1]$. \citet{Brink13} show that a solution satisfies \emph{efficiency}, \emph{linearity}, \emph{anonymity}, and \emph{weak monotonicity} if and only if it is an $\alpha$-egalitarian Shapley value for some $\alpha\in[0,1]$.
\footnote{Anonymity. For each $v \in \V^N$ and each $i \in N$, $\varphi_i(v) = \varphi_{\pi(i)}(\pi v)$, where $\pi$ is a permutation on $N$ and $\pi v$ is defined by $\pi v\left(\bigcup_{i \in S}\{\pi(i)\}\right)=v(S)$ for all $S \subseteq N$. Weak monotonicity: For each $v,w \in \V^N$ such that $v(N) \geq w(N)$ and each $i \in N$, if $v(S \cup i)-v(S) \geq w(S \cup i)-w(S)$ for all $S \subseteq N \backslash\{i\}$, then  $\varphi_i(v) \geq \varphi_i(w)$.} 
It is worth noting that if symmetry is replaced by anonymity in the statement of Theorem~\ref{thm1}, the result still holds. Consequently, the substantive difference between the two characterizations lies entirely in the fourth axiom. Weak monotonicity imposes restrictions on how payoffs may co-vary with productivity, implicitly embedding a preference for solutions that reward higher marginal contributions with higher payoffs. Null player neutrality, by contrast, involves no such cardinal comparison; it only requires that the payoff assigned to a null player be invariant to the specific null player game used for augmentation, provided that the worth of the grand coalition is held fixed. This distinction clarifies the normative gap between the two axioms and highlights the absence of any logical implication between weak monotonicity and null player neutrality. Nevertheless, Theorem~\ref{thm1} shows that replacing weak monotonicity with null-player neutrality in the characterization of \citet{Brink13} strictly enlarges the set of admissible solutions.

Theorem \ref{thm1} and Proposition 1 of \citet{Casajus13} share  the same structural backbone (efficiency, linearity, and symmetry), but differ both in the fourth axiom and in the family of rules they characterize. \citet{Casajus13} introduce the \emph{null player in a productive environment} property, which requires that a null player receive a non-negative payoff whenever the grand coalition's worth is non-negative.\footnote{Null player in a productive environment. For each $v \in \V^N$ and each $i \in N$, if $v(N) \geq 0$ and $i$ is null in $v$, then $\varphi_i(v) \geq 0$.} Under efficiency, linearity, and symmetry, this axiom pins down the family $\alpha \varphi^{ED} + (1-\alpha)\varphi^{Sh}$ with $\alpha \geq 0$. By contrast, Theorem \ref{thm1} shows that replacing null player in a productive environment by null player neutrality yields the strictly larger family $\alpha \varphi^{ED} + (1-\alpha)\varphi^{Sh}$ with $\alpha \in \mathbb{R}$, imposing no restriction on the sign of $\alpha$.

This relationship is not surprising in light of the following proposition.

\begin{proposition}\label{prop_lineanull}
If a linear solution satisfies null player in a productive environment, then it also satisfies null player neutrality. The converse does not hold. 
\begin{proof}
Let us see the implication. Let $v, w, u \in \V^N$ with $w(N) = u(N)$, and let $i \in N$ be a null player in both $w$ and $u$. By \emph{linearity}, $\varphi_i(v+w) = \varphi_i(v) + \varphi_i(w)$ and $\varphi_i(v+u) = \varphi_i(v) + \varphi_i(u)$, so it suffices to show that $\varphi_i(w-u) = 0$. Since $i$ is null in $w - u$ (as this player is null in both $w$ and $u$) and $(w-u)(N) = w(N) - u(N) = 0 \geq 0$, \emph{null player in a productive environment} yields $\varphi_i(w-u) \geq 0$. Applying the same argument to $u - w$ gives $\varphi_i(u-w) \geq 0$. Since $w - u = -(u-w)$, by \emph{linearity}, we have that $\varphi_i(w-u)=\varphi_i(-(u-w))=-\varphi_i(u-w)\leq 0$. Therefore $\varphi_i(w-u) = 0$, and \emph{null player neutrality} follows. 

The converse fails: any solution $\alpha\varphi^{ED} + (1-\alpha)\varphi^{Sh}$ with $\alpha < 0$ satisfies null player neutrality but assigns a strictly negative payoff to any null player in a game with $v(N) > 0$, violating null player in a productive environment.
\end{proof}
\end{proposition}

As the following remark shows, the relationship between the null player in a productive environment and null player neutrality established in the previous proposition breaks down in the absence of linearity.

\begin{remark}
Null player in a productive environment and null player neutrality are independent.
\begin{itemize}
   \item[(i)] The solution $\varphi$ with $\alpha=-1$ in Theorem \ref{thm1} satisfies null player neutrality but violates null player in a productive environment.
   \item[(ii)] The solution $\varphi$ such that, for each $v \in \V^N$ and each $i \in N$, $\varphi_i(v)=\max\{v(1),0\}$ clearly satisfies null player in a productive environment, but violates null player neutrality. Indeed, let $v,w,u \in \V^N$, such that $v \equiv w \equiv \mathbf{0}$, and $u(1)=1$, $u(12)=1$ and $u(S)=0$ otherwise. Player 2 is null both in $w$ and in $u$, and besides $w(N)=u(N)$. Nevertheless,
   $$
   \varphi_2(v+w)=\varphi_2(\mathbf{0})=0, \quad \text{but} \quad \varphi_2(v+u)=\varphi_2(u)=\max\{u(1),0\}=1,
   $$
\end{itemize}
\end{remark}

A further comparison connects Theorem \ref{thm1} with the characterization of the Shapley solution due to \citet{Chun91}, who shows that the Shapley solution is the unique solution satisfying efficiency, symmetry, and coalitional strategic equivalence, with no appeal to linearity. Theorem \ref{thm1} modifies this result in two respects: it replaces coalitional strategic equivalence with the strictly weaker null player neutrality, and it adds linearity as an explicit requirement. Despite the additional structure imposed by linearity, the replacement of the axiom concerning the null players enlarges the set of admissible solutions, from the Shapley solution alone to the full one-parameter family $\{\alpha\varphi^{ED}+(1-\alpha)\varphi^{Sh}:\alpha\in\mathbb{R}\}$. That is, weakening coalitional strategic equivalence to null player neutrality expands the characterized family, even in the presence of linearity.

All axioms in Theorem \ref{thm1} are necessary for the characterization.

\begin{remark}\label{independence_thm1}
Independence of axioms.
\begin{itemize}
  \item The null solution, $\varphi^0(v)=0$ for any $v \in \V^N$, satisfies linearity, symmetry and null player neutrality, but violates efficiency. 
  \item The solution $\varphi$ defined in Remark \ref{hammel} satisfies efficiency, symmetry and null player neutrality, but violates linearity.
  \item  The solution $\varphi(v)=(0,\frac{v(N)}{n-1},\cdots,\frac{v(N)}{n-1})$  satisfies efficiency, linearity and null player neutrality, but violates symmetry.
  \item The equal surplus division solution $\varphi^{ESD}$ defined as follows satisfies efficiency, symmetry and linearity, but violates null player neutrality. For each $v \in \V^N$ and each $i \in N$,
  $$
  \varphi^{ESD}_i(v)=v(i)+\frac{v(N)-\sum_{j\in N}v(j)}{n}
  $$
\end{itemize}
\end{remark}

Finally, we notice that linearity cannot be relaxed to mere additivity in our characterization. The following remark exhibits a solution that satisfies efficiency, additivity, symmetry, and null player neutrality, yet falls outside the family characterized in Theorem \ref{thm1}.

\begin{remark}\label{hammel}
Let $\mathcal{B} = (b_k)_{k \in K}$ be a Hamel basis of $\mathbb{R}$ over  $\mathbb{Q}$ with $b_1 = 1$ and $b_2 = \sqrt{2}$. Define $f: \mathbb{R} \longrightarrow \mathbb{R}$ by setting
$$
f(x) = 
    \begin{cases}
    \eta & \text{if } x \in \{1,\sqrt{2}\}, \\
    x   & \text{if } x \in \mathcal{B} \setminus \{1,\sqrt{2}\},
    \end{cases}
$$
for some $\eta > 0$, and extending to all $x \in \mathbb{R}$ by $\mathbb{Q}$-linearity,  if $x = \sum_{k \in K} \alpha_k b_k$, set $f(x) = \sum_{k \in K} \alpha_k f(b_k)$. By construction, $f$ is additive: for any $x = \sum_{k}\alpha_k b_k$ and  $y = \sum_{k}\beta_k b_k$,
\begin{align*}
    f(x + y) 
    = f\!\left(\sum_{k}(\alpha_k + \beta_k)\,b_k\right)
    = \sum_{k}(\alpha_k + \beta_k)\,f(b_k)
    = f(x) + f(y).
\end{align*}
However, $f$ is not linear over $\mathbb{R}$: indeed, $f(\sqrt{2}) = \eta \neq \sqrt{2}\,\eta = \sqrt{2}\,f(1)$ whenever $\eta \neq 0$. 

Using $f$, we define a solution $\varphi$ for TU-games as follows. For each  $T \subseteq N$ and scalar $\lambda_T \in \mathbb{R}$, set
$$
\bar{\varphi}_i(N,\lambda_T u_T) =
\begin{cases}
  f(\lambda_T) & \text{if } i \notin T, \\[6pt]
  \dfrac{\lambda_T - (n-|T|)\,f(\lambda_T)}{|T|} & \text{if } i \in T.
\end{cases}
$$
and extend to arbitrary games by additivity: if $v = \sum_{T \subseteq N}  \lambda_T u_T$ is the unanimity decomposition of $v$, define
$$
\varphi(N,v) = \sum_{T \subseteq N} \bar{\varphi}(N,\lambda_T u_T).
$$
It can be checked that $\varphi$ satisfies efficiency, symmetry and additivity. Let us see that $\varphi$ satisfies null player neutrality. Consider $v,w,u$, and let $i\in N$ be such that $i$ is null in both $w$ and $u$, with $w(N)=u(N)$. Since $i$ is null in $w$, its unanimity decomposition only involves coalitions not containing $i$, that is,
$$
w=\sum_{T\subseteq N: i\notin T}\lambda_T u_T.
$$
Hence, by construction,
$$
\varphi_i(w)=\sum_{T\subseteq N: i\notin T} f(\lambda_T).
$$
Since $f$ is additive, this can be rewritten as
$$
\varphi_i(w)=f\left(\sum_{T\subseteq N: i\notin T}\lambda_T\right).
$$
Now, for every unanimity game $u_T$ we have $u_T(N)=1$, so evaluating $w$ at the grand coalition gives
$$
w(N)=\sum_{T\subseteq N: i\notin T}\lambda_Tu_T(N)=\sum_{T\subseteq N: i\notin T}\lambda_T.
$$
Therefore,
$$
\varphi_i(w)=f(w(N)).
$$
Exactly the same argument applied to $u$ yields
$$
\varphi_i(u)=f(u(N)).
$$
Since $w(N)=u(N)$, we conclude that
$$
\varphi_i(w)=\varphi_i(u).
$$
Finally, by additivity of $\varphi$,
$$
\varphi_i(v+w)=\varphi_i(v)+\varphi_i(w)
=\varphi_i(v)+\varphi_i(u)
=\varphi_i(v+u).
$$
Thus, $\varphi$ satisfies null player neutrality. It is not linear, however, and it does not belong  to the family $\alpha\varphi^{ED} + (1-\alpha)\varphi^{Sh}$ characterized in 
Theorem~\ref{thm1}.
\end{remark}




\section{Further insights}\label{sec_further}

\citet{Brink07} proposes the notion of \emph{nullifying player} as an alternative to null player. A player $i \in N$ is considered \emph{nullifying} if participation in any coalition results in that coalition generating zero value, i.e., $v(S \cup i) = 0$ for all $S \subseteq N \setminus \{i\}$. We denote by $NfP(v)$ the set of nullifying players in $v$. Analogously to the null player property, the \emph{nullifying player property} simply states that the payoff of a nullifying player must be zero.

\textbf{Nullifying player property.} For each $v \in \V^N$ and each $i\in N$, if $i$ is nullifying then $\varphi_i(v)=0$.

Extending the parallelism between null and nullifying players, \citet{Brink07} also introduces the axiom of \emph{coalitional standard equivalence}, which mirrors coalitional strategic equivalence by replacing null players with nullifying players.

\textbf{Coalitional standard equivalence.} For each $v,w \in \V^N$ and each $i\in N$, if $i\in N$ is nullifying in $w$, then $\varphi_i(v+w)=\varphi_i(v)$.

Following the parallelism developed in Section \ref{NPN}, we define the axiom of \emph{nullifying player neutrality} by requiring that, for any two games in which player $i$ is nullifying and that have the same grand-coalition value, augmenting by either game produces the same payoff for $i$.

\noindent\textbf{Nullifying player neutrality.} For each $v,w,u \in \V^N$ and each $i\in N$, if $w(N)=u(N)$ and $i$ is nullifying in $w$ and $u$, then $\varphi_i(v+w)=\varphi_i(v+u)$.

A fundamental structural difference with the null case emerges immediately. While null player neutrality is weaker than coalitional strategic equivalence, nullifying player neutrality turns out to be equivalent to coalitional standard equivalence.

\begin{proposition}\label{equiv_nullifying}
Coalitional standard equivalence and nullifying player neutrality are equivalent.
\begin{proof}
Coalitional standard equivalence clearly implies nullifying player neutrality. For the converse, consider any $v,w\in\V^N$ such that $i\in N$ is nullifying in $w$, and let $\mathbf{0}$ denote the null game (the characteristic function identically equal to zero). Player $i$ is nullifying in $\mathbf{0}$ as well, and $w(N)=0=\mathbf{0}(N)$. By \emph{nullifying player neutrality} applied to the triple $(v,w,\mathbf{0})$, we obtain
$$
\varphi_i(v+w)=\varphi_i(v+\mathbf{0})=\varphi_i(v).
$$
Hence $\varphi$ satisfies coalitional standard equivalence.
\end{proof}
\end{proposition}

The key insight behind Proposition \ref{equiv_nullifying} is that the zero game $\mathbf{0}$ is always available as a \emph{nullifying-player game} for any player, with grand-coalition value zero. Consequently, fixing the grand-coalition value and requiring payoff invariance across nullifying-player augmentations already forces the augmentation itself to be irrelevant, recovering the full strength of coalitional standard equivalence. In the null-player case, by contrast, the zero game is of course null for every player, but it has grand-coalition worth equal to zero. Hence it cannot play the same role as in the nullifying-player case: null player neutrality allows a game $w$ in which player $i$ is null to be compared with the zero game only when $w(N)=0$. If player $i$ is null in a game $w$ with $w(N)=c\neq 0$, the axiom only compares $w$ with other games in which $i$ is null and whose grand-coalition worth is also $c$; it does not imply that adding $w$ leaves player $i$'s payoff unchanged. This asymmetry is precisely what makes null player neutrality a genuine weakening of coalitional strategic equivalence, whereas nullifying player neutrality collapses to coalitional standard equivalence.

Recall from Section \ref{NPN} that the null player property and null player neutrality are logically independent, and that the bridge between them requires linearity. For nullifying players, the analogous relationship is again different. As shown in Remark \ref{remark_NfPPandNfPN} below, the nullifying player property and nullifying player neutrality remain independent in general.

\begin{remark}\label{remark_NfPPandNfPN}
The nullifying player property and nullifying player neutrality are independent.
\begin{itemize}
\item[(i)] Consider the solution $\varphi^2$ defined as follows. For each $v \in \V^N$ and each $i\in N$,
$$
\varphi^2_i(v)=\begin{cases}
0 & \text{if }i\in \NfP(v)\text{ and }|\NfP(v)|\neq 0,\\[0.2cm]
\dfrac{v(N)}{n-|\NfP(v)|} &
\text{if }i\notin \NfP(v)\text{ and }|\NfP(v)|\neq 0,\\[0.2cm]
\varphi^{ED}_i(v) & \text{if }v(S)=v(R)\;\forall S,R (\neq \emptyset) \subseteq N,\\[0.2cm]
\varphi^{Sh}_i(v) & \text{otherwise,}
\end{cases}
$$
The solution $\varphi^2$ satisfies the nullifying player property but violates nullifying player neutrality. To see this, consider the games $v,w,u$ with $N=\{1,2,3\}$ and
\begin{center}
\begin{tabular}{cccc}
  \toprule
  $S$ & $v(S)$ & $w(S)$ & $u(S)$ \\
  \midrule
  $\varnothing$ & 0 & 0 & 0 \\
  $\{1\}$ & 1 & 0 & 0 \\
  $\{2\}$ & 1 & 0 & 1 \\
  $\{3\}$ & 1 & 0 & 1 \\
  $\{1,2\}$ & 1 & 0 & 0 \\
  $\{1,3\}$ & 1 & 0 & 0 \\
  $\{2,3\}$ & 1 & 0 & 1 \\
  $\{1,2,3\}$ & 1 & 0 & 0 \\
  \bottomrule
\end{tabular}
\end{center}
Player $1$ is nullifying in $w$ and $u$, and $w(N)=u(N)=0$. However, $\varphi^2_1(v+u)=\varphi^{Sh}_1(v+u)\neq\varphi^{ED}_1(v+w)=\varphi^2_1(v+w)$, violating nullifying player neutrality.

\item[(ii)] The solution $\varphi_i(v)=v(i)+a$ with $a\neq 0$ satisfies nullifying player neutrality but violates the nullifying player property.
\end{itemize}
\end{remark}

Although the nullifying player property and nullifying player neutrality are independent in general, under linearity they are equivalent, in contrast with the null case, where linearity yields only the one-way implication from the null player property to null player neutrality.

\begin{proposition}\label{prop_additivity_nullifying}
A linear solution satisfies the nullifying player property if and only if it satisfies nullifying player neutrality.
\begin{proof}
Linearity and the nullifying player property trivially imply nullifying player neutrality. For the converse, let $\varphi$ be a linear solution satisfying nullifying player neutrality. First, notice that \emph{linearity} implies $\varphi_i(\mathbf{0})=0$ for any $i\in N$. Now let $v\in\V^N$ and $i\in N$ nullifying in $v$. Since $v(N)=0=\mathbf{0}(N)$ (because $i$ is nullifying) and $i$ is nullifying in both $v$ and $\mathbf{0}$, \emph{nullifying player neutrality} imposes that $\varphi_i(\mathbf{0}+v)=\varphi_i(\mathbf{0}+\mathbf{0})$, so $\varphi_i(v)=\varphi_i(\mathbf{0})=0$.
\end{proof}
\end{proposition}

We are now ready to state the counterpart of Theorem \ref{thm1} for nullifying players. The result reveals a striking contrast with the null case. Whereas null player neutrality expands the admissible family to an entire one-parameter class, nullifying player neutrality uniquely pins down the equal division solution. 

Our characterization builds directly on a result due to \citet{Brink07}, which we reproduce here for ease of reference.

\begin{theorem}[\citealt{Brink07}]\label{thm_EDCSE}
A solution satisfies efficiency, symmetry, and coalitional standard equivalence if and only if it is the equal division solution, $\varphi\equiv\varphi^{ED}$.
\end{theorem}

\begin{theorem}\label{thm2}
A solution $\varphi$ satisfies efficiency, symmetry, and nullifying player neutrality if and only if it is the equal division solution, $\varphi\equiv\varphi^{ED}$.
\begin{proof}
Efficiency and symmetry follow immediately from the
definition of the equal division solution. We show that $\varphi^{ED}$ satisfies nullifying
player neutrality. Let $v,w,u\in V^N$ and let $i\in N$ be nullifying in both $w$ and $u$,
with $w(N)=u(N)$. Since the equal division solution depends only on the worth of the
grand coalition, we have
\[
\varphi^{ED}_i(v+w)=\frac{(v+w)(N)}{n}
=\frac{v(N)+w(N)}{n}
=\frac{v(N)+u(N)}{n}
=\frac{(v+u)(N)}{n}
=\varphi^{ED}_i(v+u).
\]
Thus, $\varphi^{ED}$ satisfies nullifying player neutrality.

Conversely, suppose that $\varphi$ satisfies efficiency, symmetry, and nullifying player
neutrality. By Proposition~\ref{equiv_nullifying}, nullifying player neutrality is equivalent
to coalitional standard equivalence. Hence, $\varphi$ satisfies efficiency, symmetry, and
coalitional standard equivalence. The result then follows from Theorem~\ref{thm_EDCSE}.
\end{proof}
\end{theorem}

Two points of comparison emerge between Theorem \ref{thm2} and the results reported in \citet{Brink07}.

The first one concerns Theorem~\ref{thm_EDCSE} itself, on which our proof directly builds. Both theorems use efficiency and symmetry, and differ only in the formulation of their third axiom: coalitional standard equivalence in Theorem \ref{thm_EDCSE} and nullifying player neutrality in Theorem \ref{thm2}. By Proposition \ref{equiv_nullifying} these two axioms are equivalent without any additional structural assumptions, so the two theorems are logically equivalent. The contribution of Theorem \ref{thm2} is therefore to offer an alternative formulation of the same characterization.

The second point involves the alternative characterization of the equal division solution obtained by \citet{Brink07} (Theorem~3.1 therein) using a distinct set of axioms.

\begin{theorem}[\citealt{Brink07}]\label{thm_ED}
A solution satisfies efficiency, linearity, symmetry, and the nullifying player property if and only if it is the equal division solution, $\varphi\equiv\varphi^{ED}$.
\end{theorem}

Theorem \ref{thm2} strengthens Theorem \ref{thm_ED} in two distinct respects. First, the nullifying player property is replaced by the conceptually different axiom of nullifying player neutrality. Proposition \ref{prop_additivity_nullifying} shows that the two axioms are equivalent for linear solutions, but Remark \ref{remark_NfPPandNfPN} establishes that they are logically independent in general. Second, and more substantially, Theorem \ref{thm2} dispenses with linearity. The equal division solution is uniquely characterized by efficiency, symmetry, and nullifying player neutrality alone, with no appeal to any linearity requirement. 

This stands in contrast with the null case in the previous section, where linearity plays an essential and necessary role in Theorem \ref{thm1}. The asymmetry reflects the deeper structural difference identified in Proposition \ref{equiv_nullifying}: null player neutrality is a genuine weakening of coalitional strategic equivalence, whereas nullifying player neutrality is equivalent to the full strength of coalitional standard equivalence. It is precisely this equivalence, which does not require linearity, that allows Theorem~\ref{thm2} to bypass any linearity assumption.

All axioms in Theorem \ref{thm2} are necessary for the characterization.

\begin{remark}\label{independence_thm2}
Independence of axioms in Theorem \ref{thm2}.
\begin{itemize}
\item The null solution $\varphi^0(v)=0$ for any $v \in \V^N$ satisfies symmetry and nullifying player neutrality, but violates efficiency.
\item The solution $\varphi=(0,\frac{v(N)}{n-1},\dots,\frac{v(N)}{n-1})$ satisfies efficiency and nullifying player neutrality, but violates symmetry.
\item The Shapley solution $\varphi^{Sh}$ satisfies efficiency and symmetry, but violates nullifying player neutrality.
\end{itemize}
\end{remark}

\section{Conclusions}\label{sec_conclusion}

This paper introduces and systematically analyzes the axiom of \emph{null player neutrality} in the context of TU-games. The axiom is a weakening of the classical coalitional strategic equivalence of \citet{Chun91}. Rather than requiring that augmenting a game by a null-player game leaves the null player's payoff entirely unchanged, it only requires that any payoff change be independent of which specific null-player game is used, provided the grand-coalition value is held fixed. This seemingly modest relaxation has substantial normative and structural implications.

Our main characterization (Theorem~\ref{thm1}) shows that efficiency, linearity, symmetry, and null player neutrality together pin down exactly the one-parameter family $\{\alpha\varphi^{ED} + (1-\alpha)\varphi^{Sh} : \alpha \in \mathbb{R}\}$. This result generalizes Shapley's classical characterization in a precise sense: it is obtained by replacing the null player property with the strictly weaker null player neutrality, while keeping all remaining axioms intact. The enlargement from a single solution to a one-parameter family reveals the exact normative content that the null player property adds to the structural axioms of efficiency, linearity, and symmetry, it is precisely this additional requirement that forces $\alpha = 0$ and singles out the Shapley value.

The extended family connects naturally to the existing literature. The subfamily with $\alpha \in [0,1]$ recovers the $\alpha$-egalitarian Shapley values of \citet{Joosten96}, axiomatized by \citet{Brink13} using weak monotonicity. Our Theorem~\ref{thm1} shows that replacing weak monotonicity with null player neutrality strictly enlarges the admissible set, admitting values with $\alpha < 0$ (which penalize equal division relative to Shapley) and $\alpha > 1$ (which more than fully compensate null players). Similarly, the subfamily with $\alpha \geq 0$ coincides with the family characterized by \citet{Casajus13} via the null player in a productive environment property, and our analysis reveals that, under linearity, this property implies null player neutrality but not conversely. 

The dual analysis for nullifying players, developed in Section~\ref{sec_further}, yields a striking contrast. While null player neutrality expands the admissible family to an entire one-parameter class, the analogous axiom for nullifying players collapses to the full strength of coalitional standard equivalence. As a consequence, efficiency, symmetry, and nullifying player neutrality together uniquely characterize the equal division solution (Theorem~\ref{thm2}), yielding a clean dichotomy: null player neutrality admits all real-linear combinations of the Shapley and equal division solutions, while nullifying player neutrality singles out equal division alone.

\newpage

\end{document}